\documentclass[12pt,a4paper]{article}

\usepackage{amsmath,amsthm,amssymb,amscd,a4wide}

\usepackage{graphicx}
\usepackage{caption}
\usepackage{subcaption}

\usepackage{dsfont}
\usepackage{mathrsfs}
\usepackage{setspace}
\usepackage{hyperref}

\newcommand{\ud}{\mathrm{d}}

\newcommand{\cD}{{\mathcal D}}

\DeclareMathOperator{\supp}{supp}

\numberwithin{equation}{section}

\newtheorem{theorem}{Theorem}[section]
\newtheorem{lemma}[theorem]{Lemma}
\newtheorem{prop}[theorem]{Proposition}
\newtheorem{cor}[theorem]{Corollary}
\newtheorem{remark}[theorem]{Remark}

\theoremstyle{definition}

\numberwithin{equation}{section}

\begin{document}

\thispagestyle{empty}

\vspace*{1cm}

\begin{center}

{\LARGE\bf A remark on the effect of random singular two-particle interactions} \\

\vspace*{2cm}

{\large Joachim Kerner \footnote{E-mail address: {\tt Joachim.Kerner@fernuni-hagen.de}}}%

\vspace*{5mm}

Department of Mathematics and Computer Science\\
FernUniversit\"{a}t in Hagen\\
58084 Hagen\\
Germany\\

\end{center}

\vfill

\begin{abstract} In this note we study a two-particle bound system (molecule) moving on the positive half-line $\mathbb{R}_+$ under the influence of randomly distributed singular two-particle interactions generated by a Poisson process. We give a rigorous definition of the underlying Hamiltonian and study its spectral properties. As a main result we prove that, with finite probability, the random interactions destroy the discrete part of the spectrum which is present in the free system. Most interestingly, this phenomenon is somewhat contrary to the role attributed to random interactions in the context of Anderson localisation where disorder is generally associated with a suppression of transport.
\end{abstract}

\newpage

\section{Introduction}
In this paper we are interested in an interacting quantum two-body problem with the half-line $\mathbb{R}_+$ as one-particle configuration space. The two-particle interactions shall consist of two parts, namely a deterministic binding-potential leading to a molecular-like state and a random part being associated with randomly distributed singular and spatially localised two-particle potentials. 

The singular many-particle interactions we consider were introduced in \cite{BKSingular,BKContact} in an quantum graph setting for providing a toy-model which allows to study many-particle quantum chaos. Subsequently, this type of many-particle interactions was investigated for a two-particle system on the half-line $\mathbb{R}_+$ (i.e. in a non-compact setting) in \cite{KM16,EggerKerner17}. In \cite{KM16} the authors focussed on spectral theory whereas in \cite{EggerKerner17} the aim was to derive an explicit expression for the resolvent and use this to discuss two-particle scattering in the system. Finally, in \cite{KernerMühlenbruchBound} the authors investigated the same two-particle system on the half-line but with an additional two-particle potential leading to a molecular like state, hence providing an extension of the model discussed in \cite{QUnruh}. Quite interestingly and as discussed in [Remark~3.4,\cite{KernerMühlenbruchBound}], the presence of the binding-potential leads to peculiar spectral properties. Most importantly, for the molecule with no additional (attractive) interactions, the authors proved the existence of an eigenvalue below the bottom of the essential spectrum. The existence of such an eigenvalue eigenvalue, on the other hand, is of purely quantum mechanical nature and can be attributed to the geometry of the one-particle configuration space. 

Our goal in this paper is to understand the effect of additional repulsive random two-particle interactions on the molecule. In particular, we are interested in investigating the effect on the discrete part of the spectrum which is non-trivial in the free system as described above. Note that the theory of random Schr\"{o}dinger operators has long become an important area in mathematical physics \cite{SimonBookSchrödinger,KirschInvitation}. Arguably the most important application of this theory is found in the mathematical understanding of the phenomenon of Anderson localisation \cite{AndersonLocalisatioin}. In general, Anderson localisation refers to a suppression of (electrical) transport in solids at certain energies, being due to inhomogeneities in the lattice constituting the solid. To understand this localisation phenomenon mathematically, one starts with a Hamiltonian $H_0=-\Delta+V$ on $L^2(\mathbb{R}^d)$ and some $\mathbb{Z}^d$-periodic potential $V$ respecting the lattice structure. For a large class of potentials, such a Hamiltonian has absolutely continuous spectrum with the spectrum exhibiting a band structure \cite{ShubinBerezin,stollmann2001caught}. Now, by adding a suitable repulsive random potential $V_{\omega}$ to $H_0$ one can show that the operator $H_{\omega}:=H_0+V_{\omega}$ has (almost surely) a non-trivial pure point spectrum, i.e., $\sigma_{pp}(H_\omega) \neq \emptyset$ (see \cite{stollmann2001caught} for more details). In other words, in this picture randomness leads almost surely to a non-trivial pure point spectrum and hence to localisation of the particles which consequently leads to a suppression of transport.

In some sense contrary to this we will show in this paper that there are situations where randomness might lead to an improvement of transport by leading to a destruction of the discrete part of the spectrum with finite probability. As a matter of fact, we will assume the singular two-particle interactions to be distributed on $\mathbb{R}_+$ according to a Poisson process and hence we are dealing with a high degree of randomness, suitable even to describe amorphous materials \cite{StolzPoisson,stollmann2001caught}. 

The paper is organised as follows: In Section~\ref{Sec1} we formulate the model and give a rigorous definition of the underlying self-adjoint Hamiltonian by constructing a suitable quadratic form. In Section~\ref{Sec2} we characterise the essential part of the spectrum and show that it is (almost surely) deterministic. We then focus on the discrete part of the spectrum and show, as a main result, that the discrete part of the spectrum becomes trivial with finite probability given the singular interactions are strong enough.

\section{The model}\label{Sec1}
In this paper we consider a system of two (distinguishable) particles moving on the half-line $\mathbb{R}_+=[0,\infty)$ described by the formal Hamiltonian
\begin{equation}\label{FormalHamiltonianNEW}
H=-\frac{\partial^2}{\partial x^2}-\frac{\partial^2}{\partial y^2}+\sum_{i=1}^{\infty}v_i(x,y)\left[\delta(x-a_i(\omega))+\delta(y-a_i(\omega))\right]\ + V_b(|x-y|)\ ,
\end{equation}
with binding-potential
\begin{equation}\label{BindingPotential}
V_b(|x-y|):=
\begin{cases}
0 \quad \text{if} \quad |x-y| \leq d\ , \\
\infty \quad \text{else}\ ,
\end{cases}
\end{equation}
$0 < d < \infty$ characterising the size of the molecule. The family of symmetric and non-negative two-particle interactions $\{v_i:\mathbb{R}_+\times \mathbb{R}_+ \rightarrow \mathbb{R}_+\}$ is chosen such that $\sigma_i \in L^{\infty}_{loc}(\mathbb{R}_+)$ with $\sigma_i(y):=v_i(a_i(\omega),y)$. Furthermore, the positions (atoms) $(a_i(\omega))_{i\in \mathbb{N}}$ with $a_1(\omega) < a_2(\omega)<...$ are placed on $\mathbb{R}_+$ such that $\{l_i:=a_{i}(\omega)-a_{i-1}(\omega)\}_{i \in \mathbb{N}}$, $a_0(\omega):=0$, forms a family of i.i.d random variables with $\mathbb{P}[l_i \in [a,b]]=\nu \int_{a}^{b}e^{-\nu l}\ \ud l$ on the probability space $(\Pi,\xi,\mathbb{P})$ (Poisson process), see \cite{StolzPoisson,PasturFigotin} for more details. Note that $\nu > 0$ denotes the concentration of the Poisson process.

%
%
%
%
Due to the presence of the binding-potential, the two-particle configuration space has been reduced from $\mathbb{R}^2_+$ to $\Omega$ which is given by
\begin{equation}
\Omega:=\{(x,y) \in \mathbb{R}^2_+\ | \ |x-y| \leq d\}\ .
\end{equation}
For later purposes we also define 
\begin{equation}
\Gamma_i(\omega):=\{(x,y) \in \Omega\ |\  x=a_i(\omega) \quad \lor \quad y=a_i(\omega) \}\cap\Omega\ ,
\end{equation}
and 
\begin{equation}
\partial \Omega_{D}:=\{(x,y) \in \Omega\ |\  |x-y|=d \}\ .
\end{equation}
Now, in order to arrive at a rigorous realization of \eqref{FormalHamiltonianNEW} we construct a suitable quadratic form on the Hilbert space $L^2(\Omega)$. We define
\begin{equation}\label{QuadraticForm}
q_{\omega}[\varphi]:=\int_{\Omega}|\nabla \varphi|^2 \ \mathrm{d}x +\sum_{i}\int_{\Gamma_i(\omega)} \sigma(y)|\varphi_{\bigl|\Gamma_i(\omega)}|^2(y) \ \mathrm{d}y\ ,
\end{equation} 
and note that the restrictions $\varphi_{\bigl|\Gamma_i(\omega)} \in L^2(\Gamma_i(\omega))$ are well defined according to the trace theorem for Sobolev functions \cite{Dob05}. 

Due to the infinite sum in \eqref{QuadraticForm} one has to employ different methods than those used in \cite{KM16,KernerMühlenbruchBound} for establishing a well-defined quadratic form. 
\begin{theorem} Let $(\sigma_i)_{i \in \mathbb{N}} \subset L^{\infty}_{loc}(\mathbb{R}_+)$ be given. Then, for almost every $\omega \in \Pi$, $(q_{\omega},\cD_q(\omega))$ is a non-negative, densely defined and closed quadratic form with form domain
	\begin{equation}
	\cD_q(\omega):=\{ \varphi \in H^1(\Omega): \varphi|_{\partial \Omega_{D}}=0 \ \text{and}\ \sum_{i}\int_{\Gamma_i(\omega)} \sigma_i(y)|\varphi_{\bigl|\Gamma_i(\omega)}|^2(y) \ \mathrm{d}y < \infty \}\ .
	\end{equation}
\end{theorem}
\begin{proof}
	
	
	Since $C^{\infty}_0(\Omega) \subset \cD_q(\omega)$, density follows directly. Also, non-negativity is clear from the definition of $q_\omega$ since every $\sigma_i$ is non-negative.

	 Now, let $(\varphi_n)_{n \in \mathbb{N}} \subset \cD_q(\omega)$ be a Cauchy sequence with respect to the form norm \linebreak $\|\cdot\|_{q_{\omega}}:=\sqrt{q_{\omega}(\cdot)}$ (note here that a Poincar\'{e} inequality holds on $\Omega$ due to the Dirichlet boundary conditions on $\partial \Omega_{D}$). Due to the non-negativity of $\sigma_i$ and the Poincar\'{e} inequality we observe that $(\varphi_n)_{n \in \mathbb{N}}$ is in fact also a Cauchy sequence in $H^1(\Omega)$. Denote this limit by $\varphi \in H^1(\Omega)$. Employing standard (local) trace estimates \cite{Dob05,KM16} one concludes that $\varphi|_{\partial \Omega_{D}}=0$ and 
	\begin{equation}
	\lim_{n \rightarrow \infty}\int_{\Gamma_i(\omega)} \sigma(y)|(\varphi_n)_{\bigl|\Gamma_i(\omega)}|^2(y) \ \mathrm{d}y=\int_{\Gamma_i(\omega)} \sigma(y)|\varphi_{\bigl|\Gamma_i(\omega)}|^2(y) \ \mathrm{d}y
	\end{equation}
	for each fixed $i \in \mathbb{N}$. Hence, by the Lemma of Fatou we conclude that $\varphi \in \cD_q(\omega)$ and that $\lim_{n \rightarrow \infty}q_{\omega}[\varphi_n-\varphi]=0$.
\end{proof}
According to the representation theorem for quadratic forms \cite{BEH08} there exists a unique self-adjoint operator being associated with $q_{\omega}$ for every $\omega \in \Pi$. This operator, being the Hamiltonian of our system, shall be denoted by $-\Delta^d_{\sigma}(\omega)$ and his domain by $\mathcal{D}_{\omega}(-\Delta^d_{\sigma}) \subset \mathcal{D}_q(\omega)$. 
\section{Spectral properties of $-\Delta^d_{\sigma}(\omega)$}\label{Sec2}
In this section we discuss spectral properties of the self-adjoint operator $-\Delta^d_{\sigma}(\omega)$ for values $0 < d < \infty$ and we start by characterising the essential spectrum. Recall that the form domain associated with $-\Delta^d_{\sigma}(\omega)$ is in general not independent of $\omega$ which could mean that essential spectrum depends on $\omega$ as well. However, this is (almost surely) not the case as proven in the following statement. In this context it is worth referring to a classical result which states that certain ergodicity properties of random Schr\"{o}dinger operators always lead to purely deterministic spectral parts \cite{PasturRandom,KirschErgodic,PasturFigotin}.
\begin{theorem}[Essential spectrum]\label{TheoremEssentialSpectrum} Let $(\sigma_i)_{i \in \mathbb{N}} \subset L^{\infty}_{loc}(\mathbb{R}_+)$ be given. Then 
	\begin{equation}
	\sigma_{ess}(-\Delta^d_{\sigma}(\omega))=[\pi^2 / 2d^2,\infty)
	\end{equation}
	holds almost surely.
\end{theorem}
\begin{proof} To show $[\pi^2 / 2d^2,\infty) \subset \sigma_{ess}(-\Delta^d_{\sigma}(\omega))$ we use an appropriate Weyl sequence $(\varphi_n)_{n \in \mathbb{N}} \subset \cD_q(\omega)$ [Lemma~1.4.4,\cite{stollmann2001caught}]. More explicitly, we use the same Weyl sequence as in the proof of [Theorem~3.1, \cite{KernerMühlenbruchBound}] which consists of ground states of the Laplacian on suitable rectangles subject to Dirichlet boundary conditions. Since, with probability one, there exists for any $k \in \mathbb{N}$ a number $i(k) \in \mathbb{N}$ such that $l_{i}(\omega) > k$, we can place the sequence of rectangles in $\Omega$ such that the boundary contributions in \eqref{QuadraticForm} related to the segments $\{\Gamma_i(\omega)\}$ do not contribute when evaluating $q_{\omega}[\varphi_n]$. Hence, in complete analogy to [Theorem~3.1, \cite{KernerMühlenbruchBound}] we conclude that $[\pi^2 / 2d^2,\infty) \subset \sigma_{ess}(-\Delta^d_{\sigma}(\omega))$.

Finally, in order to prove that $\inf \sigma_{ess}(-\Delta^d_{\sigma}(\omega))=\pi^2 / 2d^2$ we only have to take into account that $-\Delta^d_{\sigma\equiv 0}(\omega)$ forms a comparison operator, i.e. $-\Delta^d_{\sigma\equiv 0}(\omega) \leq -\Delta^d_{\sigma}(\omega)$, and that \linebreak $\inf \sigma_{ess}(-\Delta^d_{\sigma\equiv 0}(\omega))=\pi^2 / 2d^2$ by [Theorem~3.1,\cite{KernerMühlenbruchBound}].

\end{proof}
We now turn attention towards the discrete part of the spectrum and recall the following (deterministic) result which was established in \cite{KernerMühlenbruchBound}. 
\begin{theorem}[Discrete spectrum I] \label{FreeDiscreteSpectrum} Assume that $\sigma_i\equiv 0$ for all $i \in \mathbb{N}$. Then $\sigma_{d}(-\Delta^{d}_{\sigma}(\omega))\neq \emptyset$.
\end{theorem}
\begin{remark} Theorem~\ref{FreeDiscreteSpectrum} is interesting from a physics point of view since it implies that a molecule in a corresponding eigenstate remains spatially localised without the presence of attractive interactions. Indeed, this effect is a purely quantum-mechanical one originating from the geometry of the one-particle configuration space (the half-line).
	\end{remark}
The goal is now to understand the influence of the random singular two-particle interactions on the discrete part of the spectrum. From a physical perspective, the strong randomness (Poisson process) is assumed to enforce localisation effects \cite{stollmann2001caught}, in particular, if $\|\sigma_i|_{\Gamma_i(\omega)}\|_{\infty}$ doesn't tend to zero fast enough. In other words, the particles are hindered from escaping to infinity by an infinite number of strongly repelling potential barriers leading to backscattering. The most dramatic situation would occur if $\|\sigma_i|_{\Gamma_i(\omega)}
\|_{\infty} \rightarrow \infty$, effectively leading to Dirichlet boundary conditions along the associated segments (note that Dirichlet boundary along $\Gamma_i(\omega)$ are formally induced by setting $\sigma_i|_{\Gamma_i(\omega)}=\infty$). On the other hand, however, a large number of atoms $a_i(\omega)$ in a finite interval around zero tend to force the particles to escape to infinity since they start repelling each other.

In a first result we prove that, however strong the singular two-particle interactions are, the probability of having a non-trivial discrete part in the spectrum is always finite.
\begin{lemma}\label{PropErhaltDiskret} For all $(\sigma_i)_{i \in \mathbb{N}} \subset L^{\infty}_{loc}(\mathbb{R}_+)$ one has
	\begin{equation}
	\mathbb{P}[\sigma_{d}(-\Delta^{d}_{\sigma}(\omega))\neq \emptyset] > 0\ .
	\end{equation}
\end{lemma}
\begin{proof} Let $\varphi_0 \in H^1(\Omega)$ be the normalised eigenfunction of the Hamiltonian $-\Delta^{d}_{\sigma\equiv 0}(\omega)$ to an eigenvalue $E_0 < \frac{\pi^2}{2d^2}$ which exists due to Theorem~\ref{FreeDiscreteSpectrum}. Since the restrictions of all functions in $\varphi \in C^{\infty}_0(\mathbb{R}^2)$ onto $\Omega$ with $\varphi|_{\partial \Omega_{D}}=0$ form a dense subset of $\{\varphi \in H^1(\Omega): \varphi|_{\partial \Omega_{D}}=0\}$ (see for example \cite{Dob05}) we choose a sequence $(\varphi_n)_{n \in \mathbb{N}}$ of such functions such that $\varphi_n \rightarrow \varphi_0$ in $H^1(\Omega)$.
	
	Hence, for every $\varepsilon > 0$ there exists a number $n_0 \in \mathbb{N}$ such that
	\begin{equation}
	\left|\frac{\|\nabla \varphi_n\|^2_{L^2(\Omega)}}{\|\varphi_n\|^2_{L^2(\Omega)}}-E_0\right| < \varepsilon\ 
	\end{equation}
	for all $n > n_0$. On the other hand, since the probability that $[0,R_n]$ (with $R_n$ such that $\supp \varphi_n \subset B_{R_n}(0)$) does not contain any atoms is finite for every $n$, the statement follows by the minmax principle \cite{BEH08}.
\end{proof}
Lemma~\ref{PropErhaltDiskret} shows that even arbitrarily strong singular interactions may not destroy the discrete part of the spectrum. As illustrated by the proof, this happens whenever the atoms $(a_i(\omega))_{i \in \mathbb{N}}$ are situated at large enough distance from the origin. However, as shown by the next result, there are indeed situations in which the discrete part of the spectrum is destroyed.

\begin{theorem}\label{DiscreteSpectrumDestroyed} There exists a constant $\gamma=\gamma(d) > 0$ such that if $\inf \sigma_k> \gamma(d)$ for one $k \in \mathbb{N}$ then
	\begin{equation}
	\mathbb{P}[\sigma_{d}(-\Delta^{d}_{\sigma}(\omega))= \emptyset] > 0\ .
	\end{equation}
\end{theorem}
\begin{proof} We first describe how to determine $\gamma$: Consider the triangle $D:=\{(x,y) \in \mathbb{R}^2_+ | 0\leq x \leq d \ \text{and}\ y \leq d-x \}$. We dissect it into four parts through the lines $x=a$ and $y=a$ with $a \in \left[\frac{1}{\eta}(\frac{d}{2}-\delta),\frac{d}{2}-\delta\right]$, $\eta > 1$ and small enough $\delta > 0$, and denote the obtained domains by $\Omega_i$, $i=1,...,4$. We then consider, on each domain $\Omega_i$, the Laplacian with Robin boundary conditions of constant $\gamma > 0$ along the boundary segments formed by the lines $x=a$ and $y=a$ and Neumann boundary conditions elsewhere. We denote the corresponding operators by $(-\Delta,\gamma,\Omega_i)$ and by $\mu_i(\gamma,a) > 0$ the corresponding ground state eigenvalues. 
	
Consider $(-\Delta,\gamma,\Omega_1)$ with $\Omega_1$ being the square: By Prop.~\ref{AuxiliaryResultI} it is readily verified that $\mu_1(\gamma,a) \geq \mu_1(\gamma,\frac{d}{2}-\delta)$. On the other hand, well-known results \cite{KatoPerturbation,BruneauPopoff} allow us to establish that $\mu_1(\gamma,\frac{d}{2}-\delta) \rightarrow \mu_1(\infty,\frac{d}{2}-\delta)$ as $\gamma \rightarrow \infty$ with $\mu_1(\infty,\frac{d}{2}-\delta)$ denoting the lowest eigenvalue corresponding to the Laplacian with Dirichlet boundary conditions along the segments where previously Robin boundary conditions were imposed. Since, by a separation of variables, $\mu_1(\infty,\frac{d}{2}-\delta) \geq \frac{2\pi^2}{d^2}$ we obtain $\mu_1(\gamma,a) > \frac{\pi^2}{2d^2}$ for all $a \in \left[\frac{1}{\eta}(\frac{d}{2}-\delta),\frac{d}{2}-\delta\right]$ given $\gamma$ is large enough.
	
 Now, consider $(-\Delta,\gamma,\Omega_2)$ with $\Omega_2$ being the triangle: Again by Prop.~\ref{AuxiliaryResultI} we conclude $\mu_2(\gamma,a) \geq \mu_2(\gamma,\frac{1}{\eta}(\frac{d}{2}-\delta))$. Since $\mu_2(\infty,\frac{1}{\eta}(\frac{d}{2}-\delta))=\frac{2\pi^2}{l_{2}^2}$ with $l_2:=d\left(1-\frac{1}{\eta}\right)+\frac{2\delta}{\eta}$ we conclude that $\mu_2(\gamma,a) > \frac{\pi^2}{2d^2}$ for all $a \in \left[\frac{1}{\eta}(\frac{d}{2}-\delta),\frac{d}{2}-\delta\right]$ given $\gamma$ is large enough, $\eta \rightarrow 1$ and $\delta > 0$ small enough.
	
	Regarding $(-\Delta,\gamma,\Omega_3)$ and $(-\Delta,\gamma,\Omega_4)$ we first note that $\Omega_3$ and $\Omega_4$ are congruent and it is therefore enough to consider $\Omega_3$ only: We reflect $\Omega_3$ along the axis $y=d-x$ to obtain $\Omega^{'}_3$ and we then reflect $\Omega^{'}_3$ across the axis $y=d$ to obtain the half-cross-shaped domain $\tilde{\Omega}_3$. Obviously, one has $\mu_3(\gamma,a) \geq \tilde{\mu_3}(\gamma,a)$ where $\tilde{\mu_3}(\gamma,a)$ is the lowest eigenvalue of the Laplacian on the domain $\tilde{\Omega}_3$ with corresponding Robin boundary conditions. In a next step we realise that $\tilde{\mu_3}(\gamma,a)$ can be bounded from below by the lowest eigenvalue of the Laplacian on the rectangle $[0,a]\times [a,2d-a]$ with Robin boundary conditions along the segments where $y=a$ and $y=2d-a$ and Neumann boundary conditions elsewhere. Denote this eigenvalue by $\hat{\mu}_3(\gamma,a)$. By Prop.~\ref{AuxiliaryResultII} we get $\hat{\mu}_3(\gamma,a)\geq \frac{1}{f(\eta)}\hat{\mu}_3(\gamma,\frac{d}{2}-\delta)$ with $f(\eta) \rightarrow 1$ as $\eta \rightarrow 1$. Since $\hat{\mu}_3(\infty,\frac{d}{2}-\delta)=\frac{\pi^2}{(d+2\delta)^2}$ we obtain $\hat{\mu}_3(\gamma,a)\geq \frac{1}{f(\eta)}\frac{\pi^2}{(d+2\delta)^2}$ and consequently $\mu_3(\gamma,a)> \frac{\pi^2}{2d^2}$ for all $a \in \left[\frac{1}{\eta}(\frac{d}{2}-\delta),\frac{d}{2}-\delta\right]$, $\gamma$ large enough, $\eta \rightarrow 1$ and $\delta >0$ small enough.

	Now we are ready to prove the full statement: To do this, we first observe that $a_k(\omega) \in \left[\frac{1}{\eta}(\frac{d}{2}-\delta),\frac{d}{2}-\delta\right]$ with finite probability. Let $q_{\omega}[\cdot]$ denote the corresponding quadratic form and assume that  $\sigma_{d}(-\Delta^{d}_{\sigma}(\omega)) \neq \emptyset$, i.e., that there exists an eigenvalue $E_0 < \frac{\pi^2}{2d^2}$ with corresponding normalised eigenfunction $\varphi_0 \in \cD_{q}(\omega)$. We then consider the form $\tilde{q}_k[\cdot]:=\int_{\Omega}|\nabla \cdot|^2 \ \mathrm{d}x +\int_{\Gamma_k(\omega)} \sigma_k(y)|(\cdot)_{\bigl|\Gamma_k(\omega)}|^2(y) \ \mathrm{d}y$ which is constructed from $q_{\omega}[\cdot]$ by deleting all contributions being associated with atoms $a_i(\omega)$, $i \neq k$. Obviously, 
	\begin{equation}
\tilde{q}_k[\varphi_0] \leq q_{\omega}[\varphi_0] = E_0 < \frac{\pi^2}{2d^2}
	\end{equation}
	and hence, by the minimax principle, the operator associated with $\tilde{q}_k[\cdot]$ also has an eigenvalue below $\frac{\pi^2}{2d^2}$. Assume that $\tilde{E}_0$ is this eigenvalue with corresponding normalised eigenfunction $\tilde{\varphi}_0 \in H^1(\Omega)$. Now, as in the proof of [Theorem~3.6,\cite{KernerMühlenbruchBound}], we restrict $\tilde{\varphi}_0$ to the triangle $D$ as well as to $\Omega \setminus D$, concluding that $\tilde{\varphi}_0|_{D}$ cannot vanish identically.
	Furthermore, due to inequality [(3.7),\cite{KernerMühlenbruchBound}] we obtain, for some $l\in \{1,...,4\}$,
	\begin{equation}
	\mu_l(\gamma/2,a_k(\omega))\leq \frac{\|(\nabla \tilde{\varphi}_0)|_{\Omega_l}\|^2_{L^2(\Omega_l)}+\frac{\gamma}{2}\int_{\partial \Omega_l\cap \Gamma_k(\omega)}|\tilde{\varphi}_0|_{\partial \Omega_l}|^2 \ \ud y}{\| \tilde{\varphi}_0|_{\Omega_l}\|^2_{L^2(\Omega_l)}} \leq \tilde{E}_0 < \frac{\pi^2}{2d^2}\ .
	\end{equation}
	However, this is in contradiction with the fact that $\mu_l(\gamma/2,a_k(\omega)) > \frac{\pi^2}{2d^2}$ as shown above and hence the statement follows.
\end{proof}
\begin{cor}\label{CorollaryMainTheorem} Assume that the sequence $(\sigma_i)_{i \in \mathbb{N}} \subset L^{\infty}_{loc}(\mathbb{R}_+)$ is such that \linebreak $\limsup_{i\rightarrow \infty} \left(\inf_{y \in \mathbb{R}_+}\sigma_i(y)\right)=\infty$. Then the discrete part of the spectrum becomes trivial with finite probability.
\end{cor}
\begin{remark} Regarding Theorem~\ref{DiscreteSpectrumDestroyed} and Corollary~\ref{CorollaryMainTheorem} it is interesting to note that even strong randomness (Poisson process) and arbitrarily high barriers are not sufficient to avoid the destruction of the discrete part of the spectrum. This is insofar interesting as the spectrum becomes pure point, i.e. $\sigma(-\Delta^{d}_{\sigma}(\omega))=\sigma_{pp}(-\Delta^{d}_{\sigma}(\omega))$, when formally setting
	 $\sigma_i=\infty$ for all $i \in \mathbb{N}$.
	\end{remark}
\section*{Acknowledgements}
The authors would like to thank S.~Egger, H.~Laasri and M.~Fleermann for very helpful and interesting discussions. 

\newpage 

\begin{appendix}\section{Appendix}
We establish the following auxiliary result which implies a domain monotonicity property for the Robin Laplacian on rectangles. Note that such a monotonicity is, on general bounded domains $\Omega$, very difficult to establish already for the Neumann Laplacian \cite{Dob05}. 

For the convenience of the reader use a shorthand notation for the volume and surface integrals in the appendix.

\begin{prop}\label{AuxiliaryResultI} Consider the rectangle $\Omega_1:=[a,b] \times [c,d]$ and the scaled rectangle $\Omega_2:=[\alpha a,\alpha b] \times [\beta c,\beta d]$ with $\alpha, \beta \geq 1$. Furthermore, for $i=1,2$ we consider on $H^1(\Omega_i)$ the (Robin-)form
	\begin{equation}
	q_{i}[\varphi]:=\int_{\Omega_i} |\nabla \varphi|^2 + \sum_{j=1}^{4}\sigma_j \int_{\partial\Omega_{ij}}|\varphi|_{\partial\Omega_{ij}}|^2 \ ,
	\end{equation}
	with $\sigma_j \in \mathbb{R}_+$. Furthermore, $\partial\Omega_{ij}$ denotes the $j$-th boundary segment of $\Omega_i$. For $\varphi \in H^1(\Omega_1)$ we define a function $\tilde{\varphi} \in H^1(\Omega_2)$ by $\tilde{\varphi}(x,y):=\varphi(\frac{x}{\alpha},\frac{y}{\beta})$ and obtain
	\begin{equation}
	\frac{q_1[\varphi]}{\|\varphi\|^2_{L^2(\Omega_1)}} \geq  \frac{q_2[\tilde{\varphi}]}{\|\tilde{\varphi}\|^2_{L^2(\Omega_2)}}\ .
	\end{equation}
\end{prop}
\begin{proof} We calculate
	\begin{equation}\begin{split}
	q_1[\varphi]=&\int_{a}^{b}\int_{c}^{d}|\nabla_x \varphi(x,y)|^2 +\int_{a}^{b}\int_{c}^{d}|\nabla_y \varphi(x,y)|^2 \\
	&+\sigma_1 \int_{c}^{d}|\varphi(a,y)|^2 +\sigma_2 \int_{c}^{d}|\varphi(b,y)|^2 \\
	&+\sigma_3 \int_{a}^{b}|\varphi(y,c)|^2 +\sigma_4 \int_{a}^{b}|\varphi(y,d
	)|^2 \\
	=& \ \frac{\alpha}{\beta}\int_{\Omega_2}|\nabla_x \tilde{\varphi}(x, y)|^2+\frac{\beta}{\alpha}\int_{\Omega_2}|\nabla_y \tilde{\varphi}(x, y)|^2 \\
	&+\frac{\sigma_1}{\beta} \int_{\beta c}^{\beta d}|\tilde{\varphi}(\alpha a,y)|^2 +\frac{\sigma_2}{\beta} \int_{\beta c}^{\beta d}|\tilde{\varphi}(\alpha b,y)|^2 
	\\
	&+\frac{\sigma_3}{\alpha} \int_{\alpha a}^{\alpha b}|\tilde{\varphi}(y,\beta c)|^2 +\frac{\sigma_4}{\alpha} \int_{\alpha a}^{\alpha b}|\tilde{\varphi}(y,\beta d)|^2  \ ,
	\end{split}
	\end{equation}
	and 
	\begin{equation}\begin{split}
	\|\varphi\|^2_{L^2(\Omega_1)}&=\int_{a}^{b}\int_{c}^{d}|\varphi(x,y)|^2 \\
	&=\frac{1}{\alpha\beta}\cdot \|\tilde{\varphi}\|^2_{L^2(\Omega_2)}\ .
	\end{split}
	\end{equation}
	Since $\alpha, \beta \geq 1$ a simple estimation yields the result.
\end{proof}
In a similar way one can treat rectangles which are increased in one direction and decreased in the other direction.
\begin{prop}\label{AuxiliaryResultII}Consider the rectangle $\Omega_1:=[a,b] \times [c,d]$ and the scaled rectangle $\Omega_2:=[\alpha a,\alpha b] \times [\beta c,\beta d]$ with $\alpha \geq 1$ and $1 > \beta > \lambda > 0$. Furthermore, for $i=1,2$ and $\varphi \in H^1(\Omega_i)$ we consider the form
	\begin{equation}
	q_{i}[\varphi]:=\int_{\Omega_i} |\nabla \varphi|^2 + \sum_{j=1}^{4}\sigma_j \int_{\partial\Omega_{ij}}|\varphi|_{\partial\Omega_{ij}}|^2 \ ,
	\end{equation}
	with $\sigma_j \in \mathbb{R}_+$. For $\varphi \in H^1(\Omega_1)$ we define a function $\tilde{\varphi} \in H^1(\Omega_2)$ by $\tilde{\varphi}(x,y):=\varphi(\frac{x}{\alpha},\frac{y}{\beta})$ to obtain
	\begin{equation}
	\frac{q_1[\varphi]}{\|\varphi\|^2_{L^2(\Omega_1)}} \geq \lambda^2 \cdot \frac{q_2[\tilde{\varphi}]}{\|\tilde{\varphi}\|^2_{L^2(\Omega_2)}}\ .
	\end{equation}
	\end{prop}
\end{appendix}

{\small
	\bibliographystyle{amsalpha}
	\bibliography{Literature}}

\end{document}